\documentclass[letterpaper,10pt]{article}

\usepackage[svgnames]{xcolor}
\usepackage[numbers,longnamesfirst]{natbib}

\usepackage{xstring}
\usepackage{amsmath}
\usepackage{amssymb}
\usepackage{amsthm}
\newcommand\divides{\mathbin{|}}
\usepackage{mathtools}
\usepackage[ruled,boxed,noend,linesnumbered]{algorithm2e}
\usepackage[backref=page]{hyperref}

\usepackage{microtype}
\usepackage{libertine}
\usepackage{libertinust1math}

\usepackage{tikz}
\usetikzlibrary{positioning}

\def\titl{Sparse Polynomial Interpolation and Division\\in Soft-linear
Time}
\def\pdftitle{Sparse Polynomial Interpolation and Division in Soft-linear Time}
\def\auts{Pascal Giorgi, Bruno Grenet, Armelle Perret du Cray, Daniel S. Roche}

\hypersetup{
  pdfauthor={\auts},
  pdftitle={\pdftitle},
  colorlinks,
  linkcolor=DarkBlue,
  citecolor=DarkGreen,
  urlcolor=DarkBlue,
  bookmarksnumbered,
}

\usepackage[nameinlink,capitalise]{cleveref}

\SetAlgoVlined
\DontPrintSemicolon
\SetKwInOut{Input}{Input}
\SetKwInOut{Output}{Output}
\SetKwFor{Loop}{repeat}{}{}

\usepackage{underscore}
\newcommand{\doi}[1]{\textsc{doi}: \href{https://doi.org/#1}{#1}}
\newcommand{\arXiv}[1]{arXiv: \href{https://arxiv.org/abs/#1}{#1}}
\newcommand{\HAL}[1]{\textsc{hal}: \href{https://hal.archives-ouvertes.fr/#1}{#1}}
\newcommand{\URL}[1]{\textsc{url}: \href{#1}{#1}}

\newtheorem{theorem}{Theorem}[section]
\newtheorem{fact}[theorem]{Fact}
\newtheorem{lemma}[theorem]{Lemma}

\newtheorem{definition}[theorem]{Definition}

\crefname{fact}{Fact}{Facts}

\numberwithin{equation}{section}

\newcommand{\ZZ}{\ensuremath{\mathbb{Z}}}
\newcommand{\NN}{\ensuremath{\mathbb{N}}}
\newcommand{\GF}[1]{\ensuremath{\mathbb{F}_{#1}}}

\newcommand{\bnd}[2]{\ensuremath{#1\mathopen{}\left(#2\right)\mathclose{}}}
\newcommand{\oh}[1]{\bnd{O}{#1}}
\newcommand{\softoh}[1]{\bnd{\widetilde{O}}{#1}}
\newcommand{\ceil}[1]{\ensuremath{\left\lceil#1\right\rceil}}
\newcommand{\floor}[1]{\ensuremath{\left\lfloor#1\right\rfloor}}

\newcommand{\ee}{\ensuremath{\mathbf{e}}}
\newcommand{\xx}{\ensuremath{\mathbf{x}}}
\newcommand{\poly}{\ensuremath{\mathsf{poly}}}
\newcommand{\polylog}{\ensuremath{\mathsf{polylog}}}

\newcommand{\true}{\textsc{true}\xspace}
\newcommand{\false}{\textsc{false}\xspace}
\newcommand{\fail}{\textsc{fail}\xspace}

\def\mcinterpolate{\textsc{Interpolate\_mbb}\xspace}
\def\divisionWithSparsity{\textsc{Bounded\_sparsity\_division}\xspace}
\def\exactDivision{\textsc{Exact\_division}\xspace}

\renewcommand{\backref}[1]{Referenced on %
  \expandarg\StrCount{#1}{,}[\ncommas]%
  \ifthenelse{\ncommas = 0}{page~#1}%
  {pages~\StrBefore[\ncommas]{#1}{,}\ and\StrBehind[\ncommas]{#1}{,}}%
.}

\newcommand\auth[4]{%
  \begin{minipage}{.45\textwidth}%
  \centering
      #1\\%
      \normalsize
      #2\\%
      #3\\%
      #4
  \end{minipage}%
}

\title{\titl}
\author{%
\auth{Pascal Giorgi}{LIRMM, Univ. Montpellier, CNRS}{Montpellier, France}{pascal.giorgi@lirmm.fr}
\auth{Bruno Grenet}{LIRMM, Univ. Montpellier, CNRS}{Montpellier, France}{bruno.grenet@lirmm.fr}
\\[3em]
\auth{Armelle Perret du Cray}{LIRMM, Univ. Montpellier, CNRS}{Montpellier, France}{armelle.perret-du-cray@lirmm.fr}
\auth{Daniel S. Roche}{United States Naval Academy}{Annapolis, Maryland, U.S.A}{roche@usna.edu}
}

\begin{document}
\maketitle

\begin{abstract}
  Given a way to evaluate an unknown polynomial with integer
  coefficients, we present new
  algorithms to recover its nonzero coefficients and corresponding
  exponents.
  As an application, we adapt this interpolation algorithm to the
  problem of computing the exact quotient of two given polynomials.
  These methods are efficient in terms of the bit-length of
  the sparse representation, that is,
  the number of nonzero terms, the size of
  coefficients, the number of variables, and the logarithm of the
  degree.
  At the core of our results is a new Monte Carlo randomized algorithm
  to recover a polynomial $f(x)$ with integer coefficients given a way to evaluate
  $f(\theta) \bmod m$ for any chosen integers $\theta$ and $m$.
  This algorithm has nearly-optimal bit complexity, meaning that the
  total bit-length of the probes, as well as the computational running
  time, is softly linear (ignoring logarithmic factors)
  in the bit-length of the resulting sparse polynomial.
  To our knowledge, this is the first sparse interpolation algorithm
  with soft-linear bit complexity in the total output size. For polynomials
  with integer coefficients, the best previously known results
  have at least a cubic dependency on the bit-length of the exponents.
\end{abstract}

\section{Introduction}

\paragraph{Sparse and supersparse polynomials.}

Sparse polynomial interpolation is an important and well-studied problem
in computer algebra, with numerous connections to related problems in
signal processing and coding theory.
In our context, the task is to determine the \emph{sparse
representation} of an unknown polynomial $f\in\ZZ[x_1,\ldots,x_n]$,
which is the list of
nonzero coefficients $c_1,\ldots,c_t\in\ZZ$ and corresponding exponent
tuples
$\ee_1,\ldots,\ee_t\in\NN^n$ such that
\[f = c_1 \xx^{\ee_1} + c_2 \xx^{\ee_2} + \cdots + c_t \xx^{\ee_t}.\]
Here we use the convenient notation for each monomial
\[\xx^{\ee_i} = x_1^{e_{i,1}}x_2^{e_{i,2}}\cdots x_n^{e_{i,n}}.\]
We assume every $c_i\ne 0$ and all the $\ee_i$'s are distinct.
The number of nonzero terms in $f$, also known as the \emph{sparsity},
is written as $t = \#f$.
The bit size of the sparse representation of $f$ is $t(n\log D+\log H)$%
\footnote{Unless otherwise stated, logarithms are in base $2$; We shall also use 
base-$q$ logarithms for some prime $q$, and natural logarithms for prime-related
statements.}
with $D$ the \emph{max degree} of $f$, that is the largest exponent $e_{i,j}$,
and $H$ its \emph{height}, that is the maximum magnitude of a
coefficient%
\footnote{In this work, we do not consider the case of \emph{unbalanced}
bit lengths, where the differing sizes of each coefficient and exponent
are considered in the complexity.}.

Any sparse interpolation algorithm requires some bounds on the unknown $f$
(typically on the degree, size of coefficients, and possibly number of nonzero
terms), as well as a way to evaluate $f$. The algorithm constructs a
series of evaluation points, performs said evaluations, then
performs some computations, possibly iterating these
steps before settling on the final result.

Dense polynomial interpolation algorithms have been known for centuries
and can always recover a unique result, even if the evaluation points
are not chosen by the algorithm.
However, methods such as Lagrange interpolation
scale at least linearly with the \emph{degree} of the unknown
polynomial.
Sparse polynomial algorithms, by contrast, should scale
according to the number of nonzero terms, which in general can be much
smaller than the degree.

In fact, the degree could be
\emph{exponentially larger} than the sparse representation.
Algorithms whose cost
scales with the bit-length of the exponents, i.e., the logarithm of the
degree, are called \emph{supersparse}
or \emph{lacunary} polynomial algorithms.

\paragraph{Sparse interpolation}

Sparse interpolation has received much attention since the landmark paper
by Ben-Or and Tiwari \cite{BenorTiwari:1988}, which provides a deterministic
algorithm of complexity polynomial in $T$, $D$, $n$ for multivariate polynomials
over $\ZZ$, given a bound on $T\ge t$ as input.
This algorithm is given in the context of an unknown polynomial that a black box
allows to evaluate at any point of $\ZZ$ freely chosen by the algorithm.
Numerous extensions have been proposed~\cite{zippel90,KaltofenLakshman:1988,HuangGao19}, in
particular in order to: deal with finite fields \cite{grigoriev90,HuangRao:1999,javadi10,giesbrecht11,huang21},
avoid the bound on $t$ by \emph{early termination} techniques \cite{kl03} or extend 
the problem to the case of sparse rational functions 
\cite{KaltofenYang2007,KaltofenNehring2011,cuyt11,vdHLec2021}.
Some algorithms require the black box model to be slightly relaxed 
and allow evaluations in extension rings or quotient rings
\cite{grigoriev90,mansour95,alon95,murao96,KaltofenNehring2011,giesbrecht11,Blaser14,hl15}.

\citet{gs09} described the first algorithm for a generic ring whose
complexity is polynomial in $\log D$ (\emph{supersparse} interpolation).
Their algorithm takes as input a
\emph{straight-line program} (SLP) rather than a black box. Hence, there is no 
restriction on the evaluation domain, but the evaluation cost has to be taken into account.
Subsequent works have refined the complexity bounds of this algorithm when the ring
of coefficients is a finite field, the ring of integers or rational numbers 
\cite{ArGiRo14,agr15,HuangGao20,huang20}. The best currently known complexity
is due to Huang \cite{Huang:2019} for the interpolation of an SLP of length $L$ 
on a finite field $\GF q$ of large characteristic in $\softoh{LT \log D \log q}$ bit operations.
This complexity is however not quasi-linear in the output size due to the factor $\log D$ times $\log q$.

More details on algorithms and techniques are given in Arnold's Thesis \cite{arn16} or in the survey from van der Hoeven
and Lecerf \cite{vdHLec2019}.

In unbounded coefficient domains such as $\ZZ$,
the bit size of the values involved in the evaluation 
and computation can grow exponentially. Working with
such exponential-size integers is unrealistic and may
even make the problem trivial: the unknown polynomial $f$
can be recovered from a \emph{single evaluation} at a point 
larger than any coefficient, using the $q$-adic expansion
of the result. Hence, modular techniques are needed to get
efficient algorithms~\cite{kaltofen90,hl15}.
This motivated the definition of more general black boxes that enable to perform evaluation modulo a chosen integer $m$.

\begin{definition}\label{def:mbb}
  A \emph{modular black box} (MBB, for short) for a multivariate
  polynomial $f\in\ZZ[x_1,\ldots,x_n]$ is a function that takes any
  modulus $m\in\NN$ and $n$-tuple of evaluation points
  $(\theta_1,\ldots,\theta_n)\in \{0,1,\ldots,m-1\}^n$, and produces
  the evaluation $f(\theta_1,\ldots,\theta_n) \bmod m$.
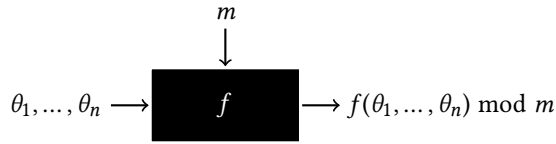
\begin{center}
\begin{tikzpicture}
\node (box) at (0,0) [color=white,fill=black,draw,very thick,minimum
width=2cm,minimum height=1cm] {$f$};
\node (a) [left=.5cm of box] {$\theta_1,\dots,\theta_n$};
\node (m) [above=.5cm of box] {$m$};
\node (f) [right=.5cm of box] {$f(\theta_1,\dots,\theta_n)\bmod m$};
\draw[thick,->] (a) -- (box);
\draw[thick,->] (m) -- (box);
\draw[thick,->] (box) -- (f);
\end{tikzpicture}
\end{center}
\end{definition}

An alternative input
for sparse interpolation
is straight-line
programs (SLP). 
An SLP naturally implements an MBB: Given the SLP for $f\in\ZZ[x_1,\ldots,x_n]$, one
can compute $f(\theta_1,\ldots,\theta_n)\bmod m$. If the SLP has length $L$, this amounts to $\oh{L}$
operations in $\ZZ/m\ZZ$, or $\softoh{L(\log m+\log H)}$ bit operations, where $H$ bounds
the absolute values of the constants used by the SLP. (More precisely, if the SLP uses $k$
constants $\le H$ in absolute value, and $H > m$, we need to reduce these $k$ integers
modulo $m$, in time $\softoh{k\log H}$.)

A fair analysis of a sparse interpolation algorithm over $\ZZ[\xx]$
should therefore consider four things: (1) the number of evaluations,
(2) the bit-length of these evaluations, (3) the arithmetic complexity
of extra processing to produce the result, and (4) the bit-length of
integers involved in the extra processing. 

\paragraph{Sparse polynomial exact division}
Another issue with sparse polynomials is the complexity of the basic 
arithmetic operations; see 
the survey of \citet{Roche2018}. 
Even for standard operations such as multiplication or division, no deterministic
quasi-linear time algorithm is known. In spite of some theoretical improvements and
practical implementations, deterministic algorithms for these operations
remain quadratic in the sparsity \cite{Johnson74,monagan07,Monagan09,monagan11,gastineau15}.
The major difficulty comes from the unpredictability of the sparsity of the result.
Quite recently, new probabilistic algorithms for sparse polynomial multiplication
have been proposed~\cite{ar15, Nakos:2020, vdH2020}. 
This led to the first quasi-linear algorithm for sparse polynomial multiplication
over the integers or finite fields with large characteristic~\cite{ggp20}, based
on sparse interpolation and sparse polynomial verification~\cite{ggp22}. 

For the Euclidean division of sparse polynomials, the case of exact division 
(when the remainder is known to be zero)
was improved by similar techniques \cite{ggp21}. This led to the first algorithm
that is quasi-linear in the sparsity, though not in the total bit size. 

\subsection{Summary of results}

We provide the first truly quasi-linear sparse interpolation algorithm, for integer polynomials.

\begin{theorem}\label{thm:Interp}
  There is a Monte Carlo randomized algorithm that,
  given an MBB for an unknown polynomial $f\in\ZZ[x_1,\ldots,x_n]$ and
    bounds $D$, $H$, and $T$ on respectively its max degree, height and sparsity,
        recovers the sparse representation of $f$
    with probability at least $\tfrac{2}{3}$.
    It requires $\oh T$ probes to the MBB plus $\softoh{T(n\log D+\log H)}$ 
    bit operations.
\end{theorem}

Based on similar techniques, we are also able to provide the first quasi-linear time algorithm for computing
the exact quotient of two sparse polynomials.

\begin{theorem}\label{thm:Div}
        There is a Monte Carlo randomized algorithm that,
    given two sparse polynomials $f$, $g\in\ZZ[x_1,\ldots,x_n]$
    such that $g$ divides $f$ and a bound $T$ on the sparsity of the quotient $f/g$,
        computes the
    sparse representation of $f/g$ with probability at least $\tfrac{2}{3}$.
    It requires $\softoh{(T+\#f+\#g)(n\log D+\log H)}$
    bit operations where $D=\deg(f)$, and $H$ is a bound on the height 
    of the three polynomials $f,g$ and $f/g$.
\end{theorem}

Our algorithms are randomized of the Monte Carlo type, meaning
that they can return incorrect results. By repeatedly running the
algorithms and taking the majority result, the probability of error
decreases exponentially in the number of iterations.

The exact division algorithm can be performed without an \emph{a priori} sparsity bound. For that, we rely on
the sparse product verification algorithm of \citet{ggp20,ggp22}. It becomes an
Atlantic City algorithm (both its correctness and running time are
probabilistic) since the verification algorithm is  randomized of Monte
Carlo type.

We present our results for multivariate polynomials but will focus on
univariate polynomials in our descriptions and proofs that follow.
This is allowed by the fairly classical
Kronecker substitution \cite{Kronecker1882,kal10a}.
Indeed, there is a one-to-one correspondence between polynomials $f\in
\ZZ[x_1,\dots,x_n]$ with $\deg_{x_i} f <D$, and univariate polynomials in
$\ZZ[x]$ of degree $<D^{n}$ through the transformation $f_u(x) =
f(x,x^D,x^{D^2},\dots,x^{D^{n-1}})$. Note that Kronecker substitution preserves
the bit size of the polynomials.
For sparse polynomials, the transformation and its inverse require
$\softoh{Tn\log D}$ bit operations.
An MBB for $f$ can simulate a univariate MBB 
for $f_u$ by evaluating $f$ at the powers of the given point. This
adds a negligible cost
in our algorithms since we probe the MBB
on points of known low order.

The rest of the paper is then devoted to univariate polynomials. By abuse of
notation we still use $D$ to denote the degree of the univariate polynomial,
instead of $D^n$.

\subsection{Main ideas}

Our new algorithms mostly combine aspects of existing techniques initiated by the work of \citet{gs09} and \citet{BenorTiwari:1988}
plus a few new techniques. We outline the most important of them to give a broad overview of the main interpolation
algorithms.

\paragraph{Finding candidate exponents}

Like in the recent line of work of Gao and Huang \cite{Huang:2019,HuangGao20,huang20,huang21}, our overall approach is to generate
\emph{candidate} terms of the unknown sparse polynomials $f$. This is achieved by interpolating $f\bmod x^p-1$ for \emph{tiny}
primes $p$, where $p \in O(T\log D)$ is so small that even performing $\softoh{p}$ operations is allowable within the targeted
complexity.

This approach originates in the work of \citet{gs09} on SLP. In that and subsequent works, the polynomial reduced
modulo $x^p-1$ is explicitly computed using dense arithmetic. This step alone is too costly to get a quasi-linear
complexity.

Our approach is to instead compute $f\bmod x^p-1$ using sparse interpolation \emph{à la} Prony. To this end, we have to
evaluate $f$ on elements of order $p$. If $\omega$ is the generator of an order-$p$ subgroup of $\GF q$,
then $f(\omega) = (f\bmod x^p-1)(\omega)$. This allows us to recover the polynomial $f$ modulo $\langle x^p-1,q\rangle$. If $\GF q$ is a
small field, namely $q \in \poly(p)$, this Prony-based interpolation has quasi-linear cost.
Since $q$ is rather small, this actually only provides the exponents modulo $p$ of $f$, but almost no information on the coefficients.

To recover the values of the coefficients, we need to work in a ring $\ZZ/m\ZZ$ for some large modulus $m$. A full Prony-based
sparse interpolation over that ring would be too expensive. However, the exponents of $f\bmod x^p-1$ have already been computed
and we only need to perform the second part of the algorithm, namely sparse interpolation with known support. Also we cannot
afford to compute a large enough prime number $m$. Instead, we work over a prime power modulus, namely $m = q^k$ for some
$k$. This part can still be done in quasi-linear time, even in this larger ring, since it amounts to solving structured linear
system of size $\oh{\#f}$.

There, we can only ensure a good probability that one-half of the terms do not collide in the reduction modulo $x^p-1$. As
proposed by \citet{Huang:2019} this can be easily turned into a Monte Carlo algorithm by doing $\oh{\log T}$ interpolations with
different primes $p$. A second problem is that, from this step, we learn only the exponents modulo $p$ and not the full exponents
themselves.  Here we can rely on the clever idea of embedding the exponents in the coefficients 
\cite{hl15,ar15,Huang:2019}.
The approach of \citet{Huang:2019} is to use the derivative for that purpose. This is well adapted for SLP since the derivative
can be computed by means of automatic differentiation. A more general way that encompasses the MBB, reminiscent of
Paillier encryption scheme~\cite{Paillier:1999}, has been proposed by \citet{ar15}. Given a modulus $m$, they consider both 
polynomials \(f(x)\) and \(f((1+m)x)\) in the ring \(\ZZ/m^2\ZZ\). Because of the identity \((1+m)^{e_i}\bmod m^2 = 1 + e_i m\),
the ratio of corresponding coefficients between these two polynomials reveals each exponent \(e_i\) modulo \(m^2\), provided that
term did not collide with any others. In our case, the modulus $m$ is $q^k$ and we actually perform the second part of the
Prony-based interpolation algorithm over $\ZZ/q^{2k}\ZZ$ to compute both $f(x)$ and $f((1+q^k)x)$ modulo $\langle x^p-1,q^{2k}\rangle$.

\paragraph{Finding rings with specified subgroups}
Our approach crucially relies on the ability of finding prime numbers $p,q$ and elements $\omega$ and $\omega_k$ such that $\omega$ and
$\omega_k$ are generators of order-$p$ subgroups of respectively \(\GF q\) and \(\ZZ/q^{2k}\ZZ\).
In particular, $p$ must divide $q-1$.  Effective versions of Dirichlet's theorem on primes in arithmetic progressions tell us
that, for a prime $p$, we can (usually) find another prime $q$ such that $p\divides (q-1)$, where $q \le O(p^6)$ is not too much
larger than $p$, see \cite{Rou85}.  This allows us to choose $q$ as a prime in the arithmetic progression $\{ap+1:a\ge1\}$ and to
set $\omega=\zeta^{(q-1)/p}$ for a random $\zeta \in \GF{q}$.
Furthermore,
one can easily construct an element $\omega_k$ of order $p$ in $\ZZ/q^{2k}\ZZ$ by lifting $\omega$ through Newton iteration. We
also demonstrate that $\omega_k$ is principal, which is a necessary condition to be able to solve our structured linear system
which is of transposed Vandermonde type.

Notice that changing the base ring is mandatory to minimize the bit complexity. Namely, the large rings have a modulus with
$O(\log D+\log H)$ bits, but we only do $\softoh{T}$ arithmetic operations in such rings. The tiny fields, by contrast, have a
modulus of only $\oh{\log(T\log DH)}$ bits, but require at most $\softoh{T\log D}$ operations.

\paragraph{Exact division}
To compute the quotient of two sparse polynomials $f$ and $g$ such that $g$
divides $f$, we adapt our interpolation techniques. To allow the evaluation of $f/g$
by evaluating both $f$ and $g$, we slightly change the values of $p$ and $q$ and ensure
that $\omega$, $\omega_k$ and their powers are not roots of $g$. The values of $p$ and $q$
do not grow too much: $p$ remains linear in the input plus the output bit size,
and $q$ polynomial in $p$. Since the height and sparsity of $f/g$ are unknown, we
must discover them during the computation. The idea is to begin with small bounds
for both and increase them when needed. For this we rely on sparse polynomial product and
modular product verification~\cite{ggp20,ggp22}. A delicate aspect is to intertwine both
bound increases.

\subsection{Outline of the paper}

We start with a preliminary section that gives few number theoretic results that are needed to prove the correctness of our
algorithms.

\cref{{sec:mc}} provides our softly linear interpolation algorithm extending further the main idea described above.  This
interpolation algorithm is re-used in~\cref{sec:divis} to provide a similar algorithm for the computation of the exact quotient of
two sparse polynomials. Moreover, we will present an unconditional algorithm that does not require any prior knowledge of the
quotient, and which has an expected softly linear running time.

\section{Number-theoretic preliminaries}\label{sec:prelim}

Our algorithms use number-theoretic results that are for many of them quite standard in the sparse interpolation literature. We
recall them in this section, in the specific form required for our proofs.  One slightly less common routine consists in computing
a primitive root of unity (PRU) of prime order $p$ in a ring $\ZZ/q^k\ZZ$ where $q = ap+1$ is also a prime number. We show how to
use Newton iteration for this purpose.

\subsection{Prime number generation}
Our algorithm first computes $f\bmod x^p-1$ where $f$ is the polynomial to be interpolated, and $p$ some random prime number.  The
goal is that not too many exponents of $f$ collide modulo $p$ to be able to recover the terms of $f$. We use a result of Arnold
and Roche~\cite{ar15}.  Note that similar results are given in other references \cite{ArGiRo14,HuangGao20}.

\begin{fact}[{\cite[Lemma~3.4]{ar15}}]\label{fact:collisions}
  Let $f$ be a $T$-sparse degree-$D$ univariate polynomial, and $p$ be a random
  prime number in $(\lambda,2\lambda)$ where $\lambda \ge
  \frac{5}{3\epsilon(1-\gamma)} (T-1)\ln D$ for some $\gamma$ and $\epsilon$.
  Then $f\bmod x^p-1$ has at least $\gamma T$ collision-free terms with
  probability at least $1-\epsilon$.
\end{fact}

To compute $f\bmod x^p-1$, one has to evaluate $f$ on $p$-PRUs. First, we need a $p$-PRU $\omega\in\GF q$ for some prime
$q$, and then a $p$-PRU $\omega_k\in\ZZ/q^k\ZZ$ for some integer $k$. To get $\omega$, we actually generate the triple
$(p,q,\omega)$ in a single algorithm, with the required properties. In particular, we need to find two prime numbers
$p$, $q$ such that $p\divides(q-1)$, that is $q$ is in the arithmetic progression $\{ap+1:a\ge 1\}$, and such that
$q=\poly(p)$. To this end, we generate $p$ at random and sample random elements $<p^6$ in the arithmetic progression
until a prime $q$ is found. Such an algorithm can be found in Arnold's Ph.D. thesis~\cite{arn16} with a rigorous proof
based on effective versions of Dirichlet's theorem~\cite{AH15,Sed18}. The next fact presents a variant with better
probability bounds and a larger range of validity. We provide the complete proof in a short note \cite{proofs}.

\begin{fact}\label{lem:mcpap}
  There exists an explicit Monte Carlo algorithm which, given a bound
  $\lambda\ge\frac{2^{58}}{\epsilon^2}$, produces a triple
  $(p,q,\omega)$ that has the following properties with probability at least
  $1-\epsilon$, and returns \fail otherwise:
  \begin{itemize}
  \item $p$ is uniformly distributed amongst the primes of $(\lambda,2\lambda)$;
  \item $q\le\lambda^6$ is a prime such that $p\divides(q-1)$;
  \item $\omega$ is a $p$-primitive root of unity in $\GF q$;
  \end{itemize}
  Its worst-case bit complexity is $\polylog(\lambda)$. Further,  if $\lambda\ge
  \sqrt[5]{\frac{48}{\epsilon}\ln K}$ for some integer $K>0$, the probability that $q$ divides $K$ is at most $\epsilon$.
\end{fact}

While the rigorous proof of this fact implies to have large values for $\lambda$, it is not too difficult to see by
running few experiments that such triples exist with good probability even for smaller values. One can find some
preliminary experiments in our short note \cite{proofs}. In this paper, we rely on \cref{lem:mcpap} to provide
rigorously proven algorithm, thus implying limitations on its practicability. Nevertheless,
\cref{algo:mcinterpolate,alg:divT} can be turned into practical ones just by ignoring the constant
$\frac{2^{58}}{\epsilon^2}$ but without any formal proof.

\subsection{Generators of prime-order subgroups}

In the crucial steps of our interpolation algorithm, we need to evaluate in a small size-$p$ multiplicative subgroup within a
larger ring of order $q^k$, where $p\divides(q-1)$ and $k\ge 1$. In order to do so, we need a generator of the order-$p$ subgroup
of the ring $\ZZ/q^k\ZZ$, that is, a $p$th primitive root of unity (PRU) in the ring.

One way to obtain such a generator would be to take a random invertible
element in the ring and raise it to the power $\varphi(q^k)/p =
(q-1)q^{k-1}/p$ modulo $q^k$. The result will certainly have
multiplicative order which divides $p$, and therefore this power of a
random element is a $p$-PRU unless it equals 1.

Unfortunately, that approach is too costly for our purposes, because the
modulus and exponent could both have roughly $k\log q$ bits. There is a
solution to this: take a $p$-PRU $\omega$ in the field $\ZZ/q\ZZ$, and
lift it to a $p$-PRU $\omega_k$ in $\ZZ/q^k\ZZ$ using a Newton iteration. This works
because of the following elementary lemma.

\begin{lemma}\label{lem:qkgen}
  Suppose $p,q$ are primes such that $p\divides(q-1)$ and $k\ge 1$.
  Let $\omega_k$ be any $p$-PRU modulo $q^k$.
  Then $\omega_k \bmod q$ is also a $p$-PRU modulo $q$. 
  Moreover, $\omega_k$ is principal, that is $\omega_k^i-1$ is not
  a zero divisor for $0 < i < p$.
\end{lemma}
\begin{proof} 
  Let $g$ be any generator of $(\ZZ/q^k\ZZ)^*$, which is cyclic since
  $q^k$ is a prime power. Then $g \bmod q$ must also be a generator
  of the smaller group $(\ZZ/q\ZZ)^*$; otherwise the set $\{g^i\bmod
  q^k\}_{i\ge 0}$ would be too small.
  Because $g$ is a generator and $\omega_k$ is a $p$-PRU modulo $q^k$, we
  can write $\omega_k = g^{i \varphi(q^k)/p}$ for some integer
  $i \in \{1,2,\ldots,p-1\}$. This means that
  $$\omega_k\bmod q = g^{i\varphi(q^k)/p} \bmod q
    = (g\bmod q)^{i(q-1)/p} \bmod q,$$
  where we use the fact that $\varphi(q^k) = (q-1)q^{k-1}$ and
  $a^q \bmod q = a$ for any integer $a$. Because $g \bmod q$ is a
  generator modulo $q$, and $1\le i \le p-1$,
  this means that $\omega_k\bmod q$ is a $p$-PRU modulo $q$.
  
  For the second part, since $\omega_k\bmod q$ is a $p$-PRU, 
  $\omega_k^i-1\bmod q\neq 0$ for $0 < i < p$. And zero divisors modulo
  $q^k$ must be multiple of $q$, since $q$ is prime.
\end{proof}

Roughly speaking, \Cref{lem:qkgen} states that there is a 1-1
correspondence between $p$-PRUs modulo $q$ and $p$-PRUs modulo $q^k$.
In particular, for any $p$-PRU $\omega$ modulo $q$, there is a unique
$p$-PRU $\omega_k$ modulo $q^k$ such that $\omega_k \bmod q = \omega$.
We construct the larger $p$-PRU $\omega_k$ through a standard Newton
iteration, solving the equation $\omega_k^p - 1 = 0$ modulo higher and
higher powers of $q$. Assuming we know $\omega_i = \omega_k\bmod q^i$
already, write $\omega_{2i} = \omega_i + a q^i$, where $a < q^i$
consists of the next $i$ base-$q$ digits of $\omega_k$. Solving the
modular equation $\omega_{2i}^p \bmod q^{2i} = 1$ gives
$$a = \left(\frac{1 - \omega_i^p \bmod q^{2i}}{q^i}\right) \omega_i
p^{-1} \bmod q^i,$$
where the fraction divided by $q^i$ is exact integer division, and the
inverse $p^{-1}$ is modulo $q^i$.

\begin{algorithm}
  \caption{\textsc{LiftPRU}}
  \label{alg:liftpru}
  \KwIn{Primes $p,q$ with $p\divides(q-1)$, a $p$-PRU $\omega\in\GF q$ and an integer $k\ge 1$}
  \KwOut{$\omega_k$, a $p$-PRU modulo $q^k$}
  \medskip

  $i \gets 1$ ; $\omega_1\gets\omega$ \;
  \While{$i < k$}{
    $a \gets \omega_i^p \bmod q^{2i}$ \;
    $a' \gets (1-a)/q^i$ using exact integer division \;
    $a'' \gets a' \omega_i p^{-1} \bmod q^i$ \;
    $\omega_{2i} \gets \omega_i + a'' q^i$ \;
    $i \gets 2i$
  }
  \KwRet{$\omega_i \bmod q^k$}
\end{algorithm}

\begin{theorem}\label{thm:liftpru}
  Provided $\omega$ is a $p$-PRU modulo $q$, 
  \Cref{alg:liftpru} returns a $p$-PRU $\omega_k$ modulo $q^k$. 
  It has bit complexity $\softoh{k \log^2 q}$.
\end{theorem}
\begin{proof}
  The loop runs $O(\log k)$ times.
  The dominating step is 
  $\omega_i^p \bmod q^{2i}$ at the last phase of the Newton iteration
  with $2i \ge k$. Because $p < q$, this gives the stated bit
  complexity.
\end{proof}

\section{Univariate Interpolation}\label{sec:mc}\label{ssec:mcinterp}

In this section, we present a Monte Carlo algorithm to interpolate a sparse polynomial given through an MBB.
Our algorithm builds on classical techniques but with the originality to use non-integral domains and not only finite fields.  We
first recall some of these techniques before describing the algorithm.

Given an MBB for $f$, we need to compute the exponents of $f\bmod x^p-1$. We note that evaluating $f$ at powers
of a $p$-th primitive root of unity ($p$-PRU) $\omega$ is equivalent to evaluating $f\bmod x^p-1$ at the same points. As in the
classical Ben-Or--Tiwari algorithm, given the sequence $f(1)$, $f(\omega)$, \dots, $f(\omega^{2T-1})$, we can compute a
degree-$\le T$ annihilator polynomial $\Lambda$ in $\softoh{T}$ operations in $\GF q$ using fast Berlekamp-Massey algorithm
\cite{Schonage:fastGCD:71,Dornstetter87}. The roots of $\Lambda$ are the $\omega^e$ where $e < p$ belongs to the support of
$f\bmod x^p-1$. In our case, $p$ is small and these exponents can be retrieved in $\softoh{p}$ arithmetic operations using
Bluestein's chirp transform~\cite{Bluestein:1970} to evaluate $\Lambda$ at $1$, $\omega$, \dots, $\omega^{p-1}$. Altogether, this
gives the following.

\begin{fact}\label{fact:BerlekampMassey}
  Given the evaluations of a $T$-sparse polynomial $f\in\GF q[x]$ at $1$, $\omega$, \dots,
  $\omega^{2T-1}$ where $\omega\in\GF q$ is a $p$-PRU, one can compute the exponents of $f\bmod
  x^p-1$ in $\softoh{T+p}$ operations in $\GF{q}$ or~$\softoh{(T+p)\log q}$ bit operations.
\end{fact}

During the algorithm, we need both to evaluate a sparse polynomial on a geometric progression and to reconstruct a sparse
polynomial from these evaluations and its exponents. If $f = \sum_{i=0}^{t-1} c_i x^{e_i}\in \GF{q}[x]$ is a sparse polynomial,
then for any $\omega$
\[\begin{pmatrix}
    1 &  \dotsb & 1 \\
    \omega^{e_0}&\dotsb&\omega^{e_{t-1}}\\
    \omega^{2e_0}&\dotsb&\omega^{2e_{t-1}}\\
    \vdots & & \vdots\\
    \omega^{(t-1)e_0}&\dotsb&\omega^{(t-1)e_{t-1}}
  \end{pmatrix}
  \begin{pmatrix} c_0\\c_1\\c_1\\\vdots\\c_{t-1}\end{pmatrix}
  =
  \begin{pmatrix} f(1) \\ f(\omega) \\ f(\omega^2)\\\vdots \\ f(\omega^{t-1})\end{pmatrix}.\]

This shows that the evaluation is a matrix-vector product and the interpolation the resolution of a linear system, where the
matrix is a transposed Vandermonde matrix.

These problems admit algorithms of complexity $\softoh{t}$ over any finite field through connections to dense polynomial arithmetic in degree $t$~\cite{KaltofenLakshman:1988,Bostan:2003}
Actually, these algorithms work for more general rings.
It is trivial for the matrix-vector product that does not require any inversion in the ring.
The resolution of the linear system requires the matrix to be invertible,
that is $\omega^{e_i}-\omega^{e_j}$ must be a unit for $i\neq j$.
This condition holds when $\omega$ is a $p$-th \emph{principal} root of unity,
that is when $\omega^p=1$ and $\omega^i-1$ is not a zero divisior for $0<i<p$.
The following fact summarizes these known results.

\begin{fact}\label{fact:vandermonde}
  Let $R$ be a ring, $f = \sum_{i=0}^{t-1} c_i x^{e_i}$ be a sparse polynomial over $R$,
  and $\omega$ a \emph{principal} $p$-th root of unity. Then
  \begin{itemize}
  \item evaluating $f\bmod x^p-1$  at $1$, $\omega$, \dots, $\omega^{t-1}$, and 
  \item retrieving the coefficients of $f\bmod x^p-1$ from its set of exponents and $f(1)$, \dots, $f(\omega^{t-1})$
\end{itemize}
can be done in $\softoh{t\log p}$ operations in $R$.
\end{fact}

We shall use these results over two rings.  First, using \cref{fact:BerlekampMassey} we perform the evaluation on powers of a $p$-PRU in $\GF q$ to recover the set of exponents modulo $p$. From these exponents, we rely on \cref{fact:vandermonde} with
a $p$-PRU $\omega_k \in \ZZ/q^k\ZZ$ to recover the polynomial modulo $x^p-1$ over the larger ring $\ZZ/q^k\ZZ$, using this time
both evaluation and interpolation. Note that $k$ is carefully chosen so that it allows to recover all the integer coefficients of
$f \bmod (x^p-1)$.
The correctness follows directly from Lemma~\ref{lem:qkgen} that shows that a $p$-PRU in $\ZZ/q^k\ZZ$ is also principal.

While we completely know $f\bmod x^p-1$, some terms of this polynomial come from collisions: That is, two (or more) distinct
monomials $c_ix^e_i$ and $c_jx^{e_j}$ from $f$ may \emph{collide} modulo $p$ and create the term $(c_i+c_j)x^{e_i\bmod p}$ in
$f\bmod x^p-1$.  We shall overcome this difficulty by a random choice of $p$ that guarantees that with good probability, not too
many terms collide. Other terms of $f\bmod x^p-1$ are \emph{collision-free}, that is of the form $c_ix^{e_i\bmod p}$. To recover
the exponent $e_i$ from these terms, we embed the exponents into its coefficients.

The idea, due to \citet{ar15}, is to compute the sparse representations of both
$f$ and $f((1+q^k)x)$, modulo $\langle x^p-1,q^{2k}\rangle$.  Since
$(1+q^k)^{e_i} = 1+e_iq^k\bmod q^{2k}$, a collision-free term $c_ix^{e_i}$ is mapped to $c_ix^{e_i\bmod p}$ in $f\bmod\langle x^p-1,q^{2k}\rangle$ and 
$c'_ix^{e_i\bmod p}$ in $f((1+q^k)x)\bmod\langle x^p-1,q^{2k}\rangle$ where $c'_i=c_i(1+e_iq^k)$. This allows us to recover both
$c_i$ and $e_i = (c'_i/c_i-1)/q^k$ as soon as $k$ is large enough. More precisely, we need $c_i$ to be a unit and representable in
$\ZZ/q^{2k}\ZZ$, and $(1+e_iq^k) \le q^{2k}$ so that the division by $q^k$ remains over the integers. That is, $q$ must
be chosen not to divide any coefficient and $k>\max(\frac{1}{2}\log_q 2H, \log_q D)$.

We note that there is no \emph{a priori} way to distinguish between collision-free terms and colliding terms. 
For some colliding terms, the recovered value of $e_i$ is clearly wrong since it is not integral or too
large, but one cannot avoid recovering unwanted terms in general. This is again taken care of through the choice of $p$, as in
\cite{Huang:2019,HuangGao19}, to avoid reconstructing too many erroneous terms.

\begin{fact}\label{fact:reconstruction}
  Given the sparse representation of $f(x)\bmod\langle x^p-1,q^{2k}\rangle$ and
  $f((1+q^k)x)\bmod \langle x^p-1,q^{2k}\rangle$ such that $q$ does not divide any
  coefficient of $f\bmod x^p-1$ and $k\ge\max(\frac{1}{2}\log_q 2H,\log_q D)$, one can
  compute a set of tentative terms of $f$, containing all the collision-free terms modulo
  $x^p-1$, in $\oh{T}$ arithmetic operations.
\end{fact}

\mcinterpolate given in Algorithm~\ref{algo:mcinterpolate} follows the idea from the three previous facts to
reach a softly-linear time complexity. 

\begin{algorithm}
\caption{\mcinterpolate}\label{algo:mcinterpolate} 
\LinesNumbered 
\Input{a polynomial $f\in\ZZ[x]$ represented by an MBB; bounds $D$, $T$ and $H$ on
respectively the degree, the sparsity and the height of $f$} 
\Output{the sparse representation of $f\in\ZZ[x]$ with probability $\ge \frac{2}{3}$;
otherwise any $T$-sparse polynomial or \fail}

\BlankLine

$f^*\gets 0$ ; $\epsilon \gets 1/(9\ceil{\log T})$\;
$\lambda\gets 
    \max\left(\frac{2^{58}}{\epsilon^2}, \frac{5}{\epsilon}(T-1)\ln D,
               \sqrt[5]{\frac{48}{\epsilon}T\ln H} \right)$\;
\tcc{Heuristically $\frac{2^{58}}{\epsilon^2}$ can be replaced by $1$,
    see discussion after \cref{lem:mcpap}.}
\While{$T \ge 1$}{
    Compute a triple $(p,q,\omega)$ such that $\omega\in\GF{q}$ is a $p$-PRU where $p$ and
    $q$ are prime numbers and $\lambda < p < 2\lambda$ using \cref{lem:mcpap}\; Evaluate
    $(f-f^*)$ at $1$, $\omega$, \dots, $\omega^{2T-1}$ and compute the exponents of
    $(f-f^*)\bmod\langle x^p-1,q\rangle$ using \cref{fact:BerlekampMassey} \;
    \label{step:BM}
    
    Compute a $p$-PRU $\omega_k\in\ZZ/q^{2k}\ZZ$ where
    $k=\lceil\max(\frac{1}{2}\log_q2H,\log_qD)\rceil$ using \cref{thm:liftpru}\; Evaluate
    $(f-f^*)$ at $1$, $\omega_k$, \dots, $\omega_k^{T-1}$ and compute the sparse
    representation of $(f-f^*)\bmod \langle x^p-1,q^{2k}\rangle$ using
    \cref{fact:vandermonde}\;\label{step:vdm}

    Perform the same step with shifted evaluation points to compute the sparse
    representation of $(f-f^*)((1+q^k)x)\bmod \langle x^p-1,
    q^{2k}\rangle$\;\label{step:vdm2}

    Compute tentative terms of $(f-f^*)$ using \cref{fact:reconstruction}\;
    \label{step:reconstruction}

    Add the tentative terms to $f^*$ ; $T \gets \floor{T/2}$ \;
}
\KwRet{$f^*$}\;
\end{algorithm}

\begin{theorem}\label{thm:MCI}
\sloppy
  Algorithm \mcinterpolate works as specified. It requires $\oh{T}$ probes to the MBB, $\softoh{T\log DH}$
  operations on integers of size $\oh{\log (T\log DH)}$, and $\softoh{T\log\log DH}$ operations on
  integers of size $\oh{\log DH}$. If the input is an SLP of length $L$ and if $H$ is also a bound on the absolute values of the
  constants of the SLP, the bit complexity of the algorithm is $\softoh{LT(\log D+\log H)}$.

  For any $\rho\ge 1$, $\oh{\rho}$ repetitions of the algorithm improve the success
  probability to $1-\frac{1}{2^\rho}$.
\end{theorem}

\begin{proof}[Correctness]
The algorithm has three sources of failure at each iteration. First, the algorithm may
fail to produce a triple $(p,q,\omega)$ satisfying the conditions. By \cref{lem:mcpap},
this probability is at most $\epsilon$. Second, the number of collisions of $(f-f^*)\bmod
x^p-1$ may be too large. \Cref{fact:collisions} and our choice of $\lambda$ guarantee
that with probability at least $1-\epsilon$, the number of collisions is at most
$\frac{1}{3}t$ where $t\le T$ is the true sparsity of $(f-f^*)$. Third, some
coefficients of $(f-f^*)\bmod x^p-1$ may vanish modulo $q$. \Cref{lem:mcpap} and our choice 
of $\lambda$ guarantee that this probability is at most $\epsilon$. Therefore, each iteration 
fails with probability at most
$3\epsilon = 1/3\ceil{\log T}$, whence the algorithm fails with probability at most
$\frac{1}{3}$.

We now prove that, assuming that none of these possible failures happens,
$f^* = f$ at the end of the algorithm. \Cref{fact:BerlekampMassey}
proves that \cref{step:BM} correctly computes the exponents of $(f-f^*)\bmod x^p-1$.
\Cref{fact:vandermonde} proves that \cref{step:vdm,step:vdm2} correctly compute the sparse
representations of $(f-f^*)\bmod\langle x^p-1,q^{2k}\rangle$ and its shifted counterpart.
Therefore, since $k$ is large enough, \cref{fact:reconstruction} ensures that
\cref{step:reconstruction} computes all the collision-free terms of $(f-f^*)$ plus some
erroneous terms. By assumption, the number of collisions of $(f-f^*)\bmod x^p-1$ is at
most $\frac{1}{3}t$. Since collisions involve at least two terms, the number of
colliding terms in $(f-f^*)\bmod x^p-1$ is at most $\frac{t}{6}$. Therefore, the tentative
terms at \cref{step:reconstruction} contain at least $\frac{2}{3}t$ correct terms and at
most $\frac{1}{6}t$ incorrect terms. In other words, the number of terms in $(f-f^*)$ at
the end of the iteration is at most $t-\frac{2}{3}t+\frac{1}{6}t = \frac{1}{2}t$. After
$\log T$ iterations, $f = f^*$.

To improve the success probability, we repeat the algorithm $48\rho/\log e$ times and
return the majority polynomial. Let $C$ be
the number of
repetitions that produce the correct polynomial. Since each repetition is correct with
probability at least $\frac{2}{3}$, $\mathbb E[C] = \frac{32\rho}{\log e}$. Therefore, by Chernoff
bound, the probability that the correct polynomial is produced by less than half of the
repetitions is $\mathrm{Pr}[C\le \frac{24\rho}{\log e}] = \mathrm{Pr}[C\le (1-\frac{1}{4})\mathbb
E[C]] \le \exp(-(\frac{1}{4})^2\mathbb E[C]/2) = \frac{1}{2^\rho}$. 
\end{proof}

\begin{proof}[Complexity]
  Each iteration require $3T$ probes to the MBB (with the current value of $T$). Hence the total number of probes is $<
  6T$. The evaluations of $f^*$ at powers of $\omega$ and $\omega_k$ require $\softoh{t\log p} = \softoh{T\log\log DH}$
  operations in $\GF q$ or $\ZZ/q^{2k}\ZZ$ by \cref{fact:vandermonde}. Apart from the evaluations, \cref{step:BM}
  requires $\softoh{p} = \softoh{T\log DH}$ operations in $\GF q$ using \cref{fact:BerlekampMassey} and
  \cref{step:vdm,step:vdm2} require $\softoh{T\log p} = \softoh{T\log\log DH}$ operations in $\ZZ/q^{2k}\ZZ$ using
  \cref{fact:vandermonde}.

The bit cost of each arithmetic operation  is $\softoh{\log q} = \softoh{\log(T\log D)+\log\log H)}$ for those in $\GF q$, and
$\softoh{k\log q} = \softoh{\log D+\log H}$ for those in $\ZZ/q^{2k}\ZZ$. If the MBB is implemented with an SLP, the overall bit
complexity, dominated by the evaluations of the SLP, is $\softoh{LT(\log D+\log H)}$.  Note that computing $p$, $q$, $\omega$ and
$\omega_k$ is cheap, since $p$, $q$ are rather small.
\end{proof}

Our algorithm is randomized of Monte Carlo type since it may return an incorrect answer, in addition to fail. To get a Las Vegas
variant, the algorithm should only be allowed to fail. 
For, we need a verification procedure that itself is a Las Vegas
algorithm. The problem to solve is then: Given an MBB for a polynomial $f$ and a sparse polynomial $f^*$, determine whether
$f = f^*$.  \Citet{Blaser:2009} provide deterministic algorithms
for this task but with polynomial, and not
quasi-linear complexity. Another approach relies on the same tools as Ben-Or--Tiwari algorithm. If both $f$ and $f^*$ have
sparsity at most $T$ and degree at most $D$, and $\omega$ is an element of order at least $D$, then $f-f^*$ vanishes on $1$,
$\omega$, \dots, $\omega^{2T-1}$ if and only if $f = f^*$ (\emph{cf.} for instance \cite{arn16}). It is deterministic as long as
an element of large order can be computed deterministically.

For a polynomial over $\ZZ$, we must evaluate $f$ and $f^*$ modulo some
integer $m$ to avoid expression swell. As before, we can produce a triple
$(p,q,\omega)$ such that $\omega$ is a $p$-PRU in $\GF q$. Since $\omega$
should have order $\ge D$, we take a random prime $p\ge D$, and $q\ge H$
so that the coefficients do not vanish modulo $q$. This can be done in time
$\polylog(D+H)$. Then, evaluating 
$f^*$ on $1$, $\omega$, \dots, $\omega^{2T-1}$ requires $2T$ probes to the MBB for $f$,
and $\oh{T\log D}$ operations in $\GF q$ for $f^*$. If $f$ is represented by an SLP of
length $L$, the bit complexity becomes $\softoh{LT\log(D+H) + T\log(D)\log(D+H)}$. Note
that this complexity is quadratic in $\log D$.

Altogether, we obtain a Las Vegas algorithm using $\oh{T}$ probes, $\oh{T\log D}$ operations in $\GF q$ and
$\polylog(D+H)$ bit operations, with a constant probability of failure. If $f$ is represented by an SLP, the bit
complexity
is $\softoh{LT\log(D+H)+T\log(D)\log(D+H)}$. Using repetition, we
obtain an algorithm that never fails, with the same \emph{expected} complexity.

It is an intriguing open question whether a quasi-linear Las Vegas algorithm exists.
In particular, can we verify an equality $f = f^*$ where $f$ is given by an SLP and $f^*$
is sparse, in quasi-linear time?

\section{Exact division}\label{sec:divis}

Given two sparse polynomials $f$ and $g$ such that $g$ divides $f$, the problem of computing $f/g$ can be seen as a sparse
interpolation of a specific SLP that has a single division.  As shown in~\citet{ggp21} some sparse interpolation algorithms can be
carefully adapted to produce division algorithms if there is no remainder. As the interpolation algorithms they rely on, these
division algorithms are not quasi-linear in the input plus the output bit-size.
In this section we show how to adapt of our quasi-linear interpolation algorithm to derive fast sparse polynomial exact
division. As a result, we obtain the first quasi-linear exact division algorithm for sparse polynomial over the integers.

There are three main difficulties in adapting our interpolation algorithm.  First, no bound is given for $\#(f/g)$ except the
potentially exponential degree one.
Second, we do not know the height of $f/g$ while
the interpolation algorithm depends on it.  Last, to evaluate the quotient $f/g$ at a root of unity $\omega$, we compute both
$f(\omega)$ and $g(\omega)$ and perform the division.  Hence, $\omega$ must not be a root of $g$. 

To overcome the first difficulty, we use the same method as~\citet{ggp20,ggp21}.  We guess a sparsity bound for the quotient,
interpolate a candidate quotient assuming the bound, and check its correctness \emph{a posteriori} with a probabilistic
verification. In case of failure we double the sparsity bound and start again.

Besides verifying products of sparse polynomials, we will also need in our algorithm an efficient verification of sparse polynomial
product modulo a binomial.  Such algorithms have been recently proposed by some of the authors in \cite{ggp22}, and we recall the
useful results below.

\begin{fact}[\citet{ggp22}]\label{fact:verif}
    There exists a Monte Carlo algorithm
        that,
    given three $t$-sparse degree-$D$ polynomials $f,g,h\in\ZZ[x]$ of height $\le H$,
    and $\rho\ge 1$, verify if $f=gh$.
    The algorithm can give a wrong answer with probability at most $\frac{1}{2^\rho}$ when $f\neq gh$. Its bit complexity is $\softoh{t(\log D+\log H+\rho)+\rho^4}$.

        There exists a Monte Carlo algorithm that similarly tests if $f = gh\bmod x^D-1$, with the same error probability and bit complexity $\softoh{t\rho\log D+t\log H+\rho^4\log^3D}$.
\end{fact}

A similar \emph{guess and check} method can be used to determine an appropriate bound for the height of the quotient: Start with a
small bound and increase it when necessary. Indeed, \cref{step:vdm} of algorithm \mcinterpolate correctly computes the polynomial
modulo $x^{p}-1$ as soon as $q^{2k}$ is greater than its height. There, verifying the sparse product modulo $x^p-1$ allows us to
determine if the bound on the height is large enough. This method is necessary as the bound we have for the height is exponential.

\begin{fact}[\citet{ggp21}]\label{fact:coeffbound}
    Let $f, g, q \in \ZZ[x]$ be three sparse polynomials such that
    $f=gq$. Then the height $H_q$ of $q$ satisfies $H_q \le (H_g+1)^{\lceil\frac{t-1}{2}\rceil}H_f$ where $t=\# q$ and $H_f$, $H_g$ are the respective heights of $f$ and $g$.
\end{fact}

For the last difficulty, we want $g(\omega)\neq 0$ for any $p$th primitive root of unity $\omega$ in $\GF{q}$. That is, we want
$g$ to be coprime with the $p$th cyclotomic polynomial $\Phi_p=\sum_{i=0}^{p-1}x^i$ in $\GF{q}[x]$.
In $\ZZ[x]$, if $p$ is a prime larger than $\#g$ such that $g \bmod x^p-1 \neq 0$, then $g$ and $\Phi_p$ are coprime.  If $p$ is
taken at random and large enough, namely $p = \Omega(\#g\log(\deg g))$, \cref{fact:collisions} ensures that $g\bmod x^p-1\neq 0$
with good probability.
Then, $g$ and $\Phi_p$ are coprime in $\GF{q}[x]$ if and only if $q$
does not divide their resultant, an integer bounded by 
$(\#g \cdot H_g)^{p-1}$ where $H_g$ is the height of $g$. 
We can therefore choose two primes $p$ and $q$ so
that $g$ and $\Phi_p$ are coprime in $\GF q[x]$ with good probability, using \cref{lem:mcpap}.

We first describe an algorithm to compute an exact quotient with a given
bound on its sparsity but no precise bound on its height.

\begin{algorithm}
    \LinesNumbered
    \caption{\divisionWithSparsity\label{alg:divT}}
  \Input{two sparse polynomials $f$, $g\in\ZZ[x]$ such that $f$ has degree $D$ and $g$ divides $f$; an integer $T$}
    \Output{$f/g$ with probability at least $\frac{2}{3}$, if $T\ge \#(f/g)$}
  \BlankLine

    $H_{max} \gets(1+H_g)^{\lceil\frac{1}{2}(T-1)\rceil}\cdot H_f$ where $H_f$, $H_g$ are the heights of $f$ and $g$ \;
    $\epsilon \gets \frac{1}{15}(\ceil{\log T} + \ceil{\log \log H_{max}})$; $C\gets H_{max}\cdot\#gH_g$\;
    $\lambda \gets \max\left(\frac{2^{58}}{\epsilon^2},\frac{5}{\epsilon}(\max(T,\#g)-1)\ln D, \sqrt[4]{\frac{96}{\epsilon}\ln C} \right)$\;
        $h\gets 0$; $H_0\gets H_g+1$\;
    \While{$T \ge 1$}{
        Compute $h_p=(f/g-h)\bmod \langle x^p-1,q^{2k}\rangle$ as in \mcinterpolate, where $\lambda<p<2\lambda$, $q\le\lambda^6$ and $k = \lceil\max(\frac{1}{2}\log_q(2H_0H_f),\log_qD)\rceil$
                \label{step:rp}\;
        Test if $f\bmod x^p-1=g\times(h_p+h) \bmod x^p-1$, with error probability $\le \frac{1}{\epsilon}$, using \cref{fact:verif}  \label{step:test}\;
    \If{the test returns \true}
    {
        Compute tentative terms of $f/g-h$\;
        Add the terms of height $\le H_{max}$ to $h$\;
        $T\gets \lfloor T/2 \rfloor$ \;
    }
    \lElse{
        $H_0 \gets H_0^2$
        }
  }
  \KwRet{$r$}
\end{algorithm}

The algorithm can return an erroneous polynomial by adding false terms. However
this polynomial cannot be much larger than the correct polynomial.

\begin{lemma}\label{lem:errorsize}
    Algorithm \divisionWithSparsity always returns a polynomial with at most
    $2T$ terms and height at most $T\cdot tH$ where $t$ and $H$ are the actual sparsity and
    height of the quotient we intend to compute.
\end{lemma}

\begin{proof}
    For the sparsity, \cref{step:rp} uses a Vandermonde system to interpolate a sparse polynomial
    of sparsity at most $T$ and cannot compute more than $T$ monomials. Therefore, as $T$ is
    divided by $2$ every time we add new terms to $h$, the result has
    at most $2T$ terms.

    For the height, only erroneous terms can have coefficients larger than $H$. However
    those terms necessarily come from collisions. Hence at each iteration, the sum
        of the erroneous terms is at most equal to the sum of the terms of $f/g-h$. Initially,
        $h = 0$ and the sum is bounded by $tH$. At each iteration, erroneous terms can at most
        double the sum. After $\ceil{\log T}$ iteration, the sum is bounded by $T\cdot tH$ and
        so is the height of $h$.
\end{proof}

\begin{theorem}\label{thm:divT}
  Algorithm \divisionWithSparsity works as specified. Its bit complexity is $\softoh{(T+\#f+\#g)(\log D+\log H)}$ where
  $D=\deg(f)$ and $H$ bounds the height of $f$, $g$ and $f/g$.

  For any $\rho\ge 1$, $\oh{\rho}$ repetitions of the algorithm improve the success probability to $1-\frac{1}{2^\rho}$.
\end{theorem}

\begin{proof}[Correctness]
  The algorithm may fail for five distinct reasons. The first three reasons are the same as in \mcinterpolate: It may fail to
  compute the triple $(p,q,\omega)$ required to compute $h_p$; The prime $p$ may cause too many collisions in
  $f/g-h\bmod (x^p-1)$; some terms of $f/g-h\bmod (x^p-1)$ may vanish modulo $q$.  The two other sources of failure are specific
  to this algorithm: One of the powers of $\omega$ or $\omega_k$ may be a root of $g$; The test at \cref{step:test} may fail to
  detect an error.

  The choice of $\lambda \ge\sqrt[4]{\frac{96}{\epsilon}\ln C}$ implies $\lambda \ge \sqrt[5]{\frac{48}{\epsilon}
    \ln(C^{2\lambda})}$. \Cref{fact:collisions,lem:mcpap} ensure that, with probability at least $1-3\epsilon$, the
  algorithm successfully produces a triple $(p,q,\omega)$ such that $p$ does not cause too many collisions and $q$ does
  not divide an unknown integer of value at most $C^p$.  
        If $p$ does not cause too many collisions, $g\bmod x^p-1\neq 0$. Since
        $\#g< p$, $g$ and $\Phi_p=\sum_{i=0}^{p-1}x^i$ are coprime in
        $\ZZ[x]$.  The resultant of $g$ and $\Phi_p$ is at most $(\#g H_g)^p$.
        Moreover, since $H_{max}$ bounds the height of both $h$ and $f/g$ using
        \cref{fact:coeffbound}, and since $p > T$, the height of $(f/g-h)\bmod
        x^p-1$ is at most $H_{max}^p$. Hence with probability at least
        $1-\epsilon$, $q$ does not divide the resultant of $g$ and $\Phi_p$ nor
        any coefficient of $(f/g-h)\bmod x^p-1$. In particular, $g$ and
    $\Phi_p$ remain coprime in $\GF q$ and so in $\ZZ/q^{2k}\ZZ$ since $p$-PRU
    in $\ZZ/q^{2k}$ are also $p$-PRU in $\GF q$.
        
        Altogether, the four following properties hold with probability at least $1-4\epsilon$:
        The algorithm
        succeeds in producing two primes $p$, $q$ and $\omega\in\GF q$;  
    $g$ and $\Phi_p$ are coprime in $\GF q[x]$ and in $\ZZ/q^{2k}\ZZ$; 
        There are few collisions in $f/g-h$ modulo $x^p-1$; $q$ does not divide any of the
        coefficients of $(f/g-h) \bmod x^p-1$.

        If all these conditions hold, we can use Facts~\ref{fact:BerlekampMassey}
        and~\ref{fact:vandermonde} to compute $h_p$. The choice of $k$ implies
        that $q^{2k}$ is larger than twice the height of $f/g-h$ as soon as
        $H_0$ is larger than the (unknown) height $H$ of $f/g$. In that case,
        the equality $h_p = f/g-h$ holds in $\ZZ[x]$ and the test at
        \cref{step:test} returns \true. Computing tentative terms and updating
        $h$ can then be done exactly as in \mcinterpolate. 
    
        If $H_0 < H$, there are two possibilities. Either $h_p\neq f/g-h \bmod
        x^p-1$ in $\ZZ[x]$. With probability at least $1-\epsilon$, the test
        detects that and $H_0$ is squared. Or the equality indeed holds. This
        means that the terms of $f/g-h$ that have a larger height collide
        modulo $x^p-1$. Hence, the collision-free terms are correctly computed.

        Consequently, the loop works correctly with probability $1-5\epsilon$:
        Either the number of terms that remain to be computed is halved, or the
        height bound is squared if it was too small. 
        At most $\ceil{\log\log H} \le \ceil{\log\log H_{max}}$ iterations
        where the test returns \false are needed to get to a correct bound
        $H_0\ge H$,
        and at most $\ceil{\log T}$ iterations where the test returns \true are
        needed to 
        to compute all the coefficients.
        Therefore the algorithm performs at most $(\ceil {\log T} + \ceil{\log
        \log H_{max}})$ iterations.  Its success probability is at least
        $1-5\epsilon(\ceil {\log T} + \ceil{\log \log H_{max}})\ge
        \frac{2}{3}$.
        To improve the success probability, we repeat the algorithm
        $48\rho/\log e$ times and return the majority polynomial, as in
        \mcinterpolate.
\end{proof}

\begin{proof}[Complexity]
    Since the number of iterations is logarithmic in the input and output size, the complexity of the
    algorithm is given by the complexity of one iteration. 
    As in \mcinterpolate, the algorithm requires $\softoh{T+p}$ operations in $\GF q$ 
    and $\softoh{T\log p}$ operations $\ZZ/q^{2k}\ZZ$ for  the evaluations of $h$, computing the exponents
    modulo $p$ and then retrieving the coefficients and the entire exponents.
        The evaluations of $f/g$ require $\softoh{(T+\#f+\#g)\log p}$ operations
    in both domains by \cref{fact:vandermonde} plus $\oh{\#f+\#g}$ operations in $\ZZ$ to
    reduce the initial coefficients and degree.
    As the height of an erroneous answer is at most $T^2H$ by \cref{lem:errorsize}, the maximal value of $q^{2k}$
    is $\oh{T^2H+D}$. Therefore arithmetic operations in $\ZZ/q^{2k}\ZZ$ have bit cost $\softoh{\log H+\log D}$.
    Moreover the choice of $\lambda$ ensures that $p=\softoh{(T+\#g)(\log D+\log H)}$.
    As $q$ is polynomial in $p$ this leads to a total bit complexity of $\softoh{(T+\#f+\#g)(\log D+\log H)}$.
\end{proof}

Our main division algorithm uses \divisionWithSparsity with growing sparsity bound until a result is found.

\begin{algorithm}
  \LinesNumbered
    \caption{\exactDivision\label{alg:div}}
  \Input{$f$, $g\in\ZZ[x]$, such that $g$ divides $f$, $\rho\ge 1$}
    \Output{$f/g$ with probability at least $1-\frac{1}{2^{\rho+1}}$}
  \BlankLine

    $T\gets 1$\;
    \While{\true}{
        $T\gets 2T$\;
    Compute  $\oh{\rho}$  candidates $h$ for $f/g$ using \cref{alg:divT} with sparsity bound $T$ and keep the most
    frequent one\;

    Test if $f = gh$ using the algorithm from \cref{fact:verif},  setting its failure probability to $\tfrac{1}{2^{\rho+1}T}$ \;
    If the test returns \true, \KwRet{$h$}\;
        }
\end{algorithm}

\begin{theorem}\label{thm:exactdiv} 
  Let $f$, $g$ be sparse polynomials in $\ZZ[x]$ such that $g$ divides $f$, $H$ be a bound on the height of $f,g$ and
  $f/g$, and $\rho\ge 1$. With probability at least $1-\frac{1}{2^\rho}$, Algorithm \exactDivision returns $f/g$ in
  $\softoh{(\#(f/g)+\#f+\#g)(\log D+\log H+\rho)+\rho^4}$ bit operations.

\end{theorem} 
 
\begin{proof}
  The probability $1-\frac{1}{2^\rho}$ concerns both the correctness and the complexity of the algorithm. We prove that
  each of them holds independently with probability $\ge 1-\frac{1}{2^{\rho+1}}$.

  The algorithm is incorrect when $f\neq gh$. This happens if at some iteration, the candidate
    quotient is incorrect but the verification algorithm fails to detect it.  
  Since each verification fails with probability at most $\frac{1}{2^{\rho+1}T}$ and values of $T$ range over powers of two,
  the algorithm is correct with probability at least  $1-\frac{1}{2^{\rho+1}}$.

  For the complexity we first need to bound the number of iterations. Since the values of $T$ are powers of two, the
  first value $\ge \#(f/g)$ is at most $2\#(f/g)$. As soon as $T$ reaches this value, the return value is actually $f/g$
  with probability at least $1-\frac{1}{2^{\rho+1}}$ according to \cref{thm:divT} when the number of candidates is $\geq
  48(\rho+1)/\log e$. In that case, the test which is only one-sided error, succeeds and the algorithm returns $h=f/g$.
  That is, with probability at least $1-\frac{1}{2^{\rho+1}}$, the number of iterations is $\oh{\log \#(f/g)}$. Even
  with false sparsity, \cref{lem:errorsize} ensures that the size of the candidate quotients is at most quasi-linear in
  the size of the actual quotient. Therefore we can apply \cref{thm:divT} to obtain the claimed complexity with
  probability at least $1-\frac{1}{2^{\rho+1}}$.
\end{proof}

\section*{Acknowledgements}
We are grateful to the reviewers for their insightful comments.


\newcommand{\Gathen}{\relax}\newcommand{\Hoeven}{\relax}

\end{document}